\documentclass[12pt]{article}
\usepackage{amsfonts,amssymb,amsbsy,amsmath,bm,times,t1enc,url}

\ifx\pdfoutput\undefined 
\usepackage[dvips]{graphicx}
\RequirePackage[dvips]{color}
\else
\usepackage[pdftex]{graphicx}
\RequirePackage[pdftex]{color}
\fi

\definecolor{darkgreen}{rgb}{0,0.3,0}

\newcommand{\pr}[1]{\mathbb{P}\hspace{-0.1em}\left(#1\right)}
\newcommand{\cPr}[2]{\mathbb{P}\hspace{-0.6mm}\left[\left.#1\,\right|#2\right]}
\newcommand{\E}[1]{\mathbb{E}\set{#1}}

\newcommand{\set}[1]{\left\{{#1}\right\}}
\newcommand{\eset}[2]{\left\{{#1} : \: {#2}\right\}}
\newcommand{\indic}[1]{1_{\{#1\}}}
\newcommand{\Nat}{\mathbb{N}}
\renewcommand{\Re}{\mathbb{R}}
\newcommand{\D}{{\,\mathrm{d}}}

\newcommand{\bmat}[1]{ \begin{bmatrix} #1 \end{bmatrix}}
\newcommand{\smat}[1]{ \left[\begin{smallmatrix} #1 \end{smallmatrix}\right]}
\renewcommand{\vec}[1]{{\boldsymbol{#1}}}
\newcommand{\mat}[1]{{\boldsymbol{#1}}}
\newcommand{\itext}[1]{{\qquad\text{#1}\qquad}}
\newcommand{\norm}[1]{\Vert #1 \Vert}
\newcommand{\abs}[1]{\left| #1 \right|}
\def\f{\vec{f}}
\def\g{\vec{g}} 
\def\h{\vec{h}} 
\def\p{\vec{p}} 
\def\q{\vec{q}} 
\def\u{\vec{u}} 
\def\v{\vec{v}}
\def\x{\vec{x}}
\def\A{\mat{A}} 
\def\V{\mat {V}}

\newtheorem{proposition}{Proposition}[section]

\newtheorem{conjecture}[proposition]{Conjecture}

\newtheorem{lemma}[proposition]{Lemma}

\newenvironment{proof}{{\it
Proof}\quad}{\nopagebreak\hspace*{\fill}$\square$\par
\bigbreak}

\begin{document}

\title{On the stability of two-chunk file-sharing systems}
\author{Ilkka Norros and Hannu Reittu\\
VTT Technical Research Centre of Finland\and Timo Eirola\\
Helsinki University of Technology}
\date{}

\maketitle

\begin{abstract}
We consider five different peer-to-peer file sharing systems with two
chunks, with the aim of finding chunk selection algorithms that have
provably stable performance with any input rate and assuming
non-altruistic peers who leave the system immediately after
downloading the second chunk. We show that many algorithms that first
looked promising lead to unstable or oscillating behavior. However, we
end up with a system with desirable properties. Most of our rigorous
results concern the corresponding deterministic large system
limits, but in two simplest cases we provide proofs for the stochastic
systems also. 
\end{abstract}


\section{Introduction}\label{intro} 

We consider an open network, with constant rate of incoming 'peers'. A
peer is assumed to be able to contact and communicate with any other
peer in the system (technically, this can be realised by an overlay
network built upon the Internet, where the knowledge of a peer's
IP-address enables communication with it). We assume that one special,
persistent peer, the 'seed', holds a file and wishes to distribute it
to all the others. The most effective way of doing this, in particular
when the number of peers is very large, is that as soon as a peer
receives the file, it becomes a seed itself. The number of copies of
the file then grow exponentially. To enhance performance, the file is
divided into small chunks that are spread in similar fashion so that
parts of the file may start to be multiplied before the original seed
has even once uploaded the whole file. This technique was introduced
by B.\ Cohen with his BitTorrent \cite{bittorrent} protocol. It became
soon the dominant principle of sharing large files (e.g., movies) with
peer-to-peer networking.

Moreover, the peers can be 'non-altruistic' in the sense that they
leave the system immediately having downloaded the whole file, without
necessarily slowing the system performance. It is remarkable that if
there are more than one chunks, it seems that {\em any} arrival rate
of new peers can still be sustained. However, as we shall see, some
additional algorithms are then needed for stable performance. This
paper analyses several such algorithms in the simplest relevant case
of two chunks. In real systems the number of chunks is large, e.g.\ of
the order $10^3$. However, even the case of two chunks, considered
here, is highly non-trivial. In fact, it may also be the hardest case
as regards stability. For {\em proving} stability it is anyway the
easiest case, and hopefully the solutions for the two chunk case
appear useful in more general models as well.

We work on a fully distributed scenario, relying on randomness: each
peer contacts another, uniformly randomly chosen peer, according to a
standard Poisson process, and gets to know what chunks the latter
possesses. What follows, depends on the particular
algorithm. Massoulie and Vojnovic \cite{massoulievojnovic05} were the
first to propose this model and to obtain rigorous mathematical
results on it. They allowed an arbitrary number of chunks and analysed
the corresponding deterministic large system limit (see Section
\ref{prelisec}) with the following remarkable result: if each new peer
arriving in the system obtains a roughly uniformly random chunk, the
system is stable even if all the remaining chunks are downloaded
randomly. (In the special case of two chunks, it is noted in
\cite{massoulievojnovic05} that the scheme does not give a unique rest
point in the large system limit, but for any larger number of chunks
it does.) However, the results of \cite{massoulievojnovic05} prompt
for further research in at least two directions. First, the scheme where
the seed gives a uniformly distributed chunk to every new peer makes
the seed a potential bottleneck --- thus, this algorithm is not as
fully distributed as it could be. Second, the stability of the large
system limit does not automatically guarantee the stability of the
original random system.

We focus on entirely distributed solutions, where also the first chunk
must be found randomly from the peer population, and assume
'non-altruistic' peers. In our first paper \cite{norprarei06} we
considered the so-called flash-crowd scenario where a large number of
peers arrive simultaneously but none afterwards. It was noticed that
the first phase of the copying process is asymptotically (with
increasing number of peers) equivalent to P\'{o}lya's urn model, which
is well known to converge to a random proportion of each chunk in the
system. This imbalance leads to the 'rare chunk phenomenon': one of
the chunks is not able to become common and as a result forms a
bottleneck of performance (see also \cite{reittunorrosspaswin07}). In
an open variant of this setup, with continuously incoming peers, this
could lead to instability: the number of peers in the system could
grow unboundedly, since more and more peers would be searching for the
one rare chunk. (BitTorrent counteracts to the rare chunk phenomenon
by its 'Rarest First' principle, and our last two algorithms can be
seen as distributed ways to implement this principle using very
coarse rarity estimates.) The central question of this paper is: is it
possible to avoid the severely imbalanced chunk distribution, implying
instability, without using centralized coordination of downloads?

One source of ideas for this is provided by the wide literature on urn
models (originally often related to physics; for a recent review, see
\cite{pemantle07}). For example, Ehrenfests' urn model gives an almost
ideal balance in a closed system. Still more relevant is the so-called
Friedman urn, with an analogous result for an open system, with a flow
of incoming particles. We showed in \cite{norprarei06} that if an
empty node first contacts a node having chunk $0$ ($1$) but then
downloads the opposite chunk $1$ ($0$) first (neglecting for a while
the question how a peer with that chunk could be found), the
distribution of chunks converges almost surely to $\frac12-\frac12$ as
the number of peers goes to infinity.

In this paper, we analyse five two-chunk models: (i) the Plain Random
Contact system, which is found to be very unstable; (ii) the
Deterministic First Chunk system, proposed in \cite{norprarei06},
but found unstable in the present scenario; (iii) the ideal Friedman system
(non-implementable in a distributed way in our scenario), which is
proven to be stable; and two distributed algorithms that try to
approximate the Friedman system: (iv) the Delayed Friedman system,
which may be stable but oscillates heavily, and, finally, (v) the Enforced Friedman
system which seems to provide the desired performance. 

The paper is structured as follows: general definitions and some
preliminaries are given in Section \ref{prelisec}, and the five models
are studied in Section \ref{modelsec}, each in its own
subsection. Some concluding remarks are made in Section \ref{conclsec}.

\section{Definitions and preliminaries}\label{prelisec}

We study time-homogeneous continuous time Markov processes
$S=(S_t)_{t\ge0}$ with state space $\Nat^d$, where $d$ is 2,3 or 4,
depending on the particular model. Denoting the $i$th unit vector by
$e_i$, $i=1,\ldots,d$, the transitions are always of one of the three
forms
$$
s\to s+e_i,\quad s\to s-e_i,\quad s\to s+e_i-e_j.
$$
Denote the transition intensity from state $\bm{m}$ to state $\bm{n}$
by $q(\bm{m},\bm{n})$. The process $S$ is thought as a model of a
queueing network, where the state component $i$ presents the number of
customers in network node $i$, $i=1,\ldots,d$.

The Markov process $S$ is called {\em stable}, if it is irreducible
and positively recurrent. This is equivalent to the existence of a
unique stationary probability measure. Assuming irreducibility and
finiteness of transition graph neighborhoods, stability is equivalent
to the existence of a finite set of states $C\subset\Nat^d$ such that
with any starting point $S_0$, the process reaches $C$ in a time with
finite expectation.

Let $F:(0,\infty)^d\to\Re^d$ be a Lipschitz continuous function so that
the autonomous ordinary differential equation
\begin{equation}
\label{gendynsys}
\dot{s}=F(s)
\end{equation}
has a solution $s(t)$, $t\in[0,T)$, $T\in(0,\infty]$, for every
starting point $s(0)\in(0,\infty)^d$. We say that the dynamical system
(\ref{gendynsys}) is the {\em large system limit} of the Markov process
$S$, if 
$$
F_i(x)=\lim_{N\to\infty}\left(
q(\lfloor Nx\rfloor,\lfloor Nx\rfloor+e_i)
+\frac{1}{N}\sum_{v\not\in\set{e_1,\ldots,e_d}}v_i
q(\lfloor Nx\rfloor,\lfloor Nx\rfloor+v)\right),\quad i=1,\ldots,d,
$$
where $\lfloor y\rfloor$ denotes the largest integer less than or
equal to $y$ and is defined componentwise for vectors, and
$v=(v_1,\ldots,v_d)$ runs over the different possible transition
vectors. The idea is to scale the arrival rates to the system
(transitions $v=e_i$) as well as the states by $N$, divide by $N$ and
take the limit. Thus, we assume the internal transition and exit rates
to be linear in $N$ as functions of the state. Conditions of limit
theorems showing the convergence of the Markov process towards a
deterministic limit system in such scaling have been established by
Kurtz \cite{kurtz81}. This type of results, however, tell nothing
about the stability of a stochastic system with finite $N$, and
therefore we don't review them closer here. 

The dynamical system (\ref{gendynsys}) is called {\em locally
asymptotically stable} around an equilibrium state $s^*$ (that is, a
state with $F(s^*)=0$), if there exists an open set $U$ containing
$s^*$ such that $\lim_{t\to\infty}s(t)=s^*$ for any initial state
$s(0)\in U$. The system is called {\em globally asymptotically
stable}, if it has a unique equilibrium state $s^*$, such that
$\lim_{t\to\infty}s(t)=s^*$ for any initial state
$s(0)\in(0,\infty)^d$. Finally, we call the system (\ref{gendynsys})
{\em (globally) stable}, if there is a compact set
$K\subset(0,\infty)^d$ such that for any initial state
$s(0)\in(0,\infty)^d$ the system reaches $K$ and eventually stays in $K$.

The stability of a large system limit is not known to be sufficient
nor necessary for the stability of the original Markov process. We are
not aware of any rigorous results concerning this question, but the
following remarks illuminate the difficulties in relating the two
notions. First, assume that the large system limit exists and is
globally asymptotically stable. When $d\ge3$, the trajectories can
however be very complicated and the convergence toward equilibrium
very slow. A stochastic system, how well ever fitted to the continuous
state space, evolves in jumps and does not follow any trajectory ---
only its local drift is in the best case close to the derivative of a
trajectory passing the same point. Thus, it is hard to imagine how the
stability of the stochastic system could be deduced without more
specific assumptions. Second, a dynamic system can escape to infinity
along a single trajectory (say, along the diagonal) while all other
trajectories end to a compact set. In such a case, the stochastic
systems with all $N$ could however be stable, since randomness forces
them to deviate from the transient trajectory.

However, these circumstances often coincide, and, as we shall see here
also, proving the stability of the dynamical system is usually much
easier than proving the stability of the Markov process. Therefore it
is interesting to consider the large system limits together with the
original random systems.

All the systems studied in this paper possess differentiable large
system limits, and the existence of a unique solution from any
starting point is thus always granted. They are, however, non-linear,
and proving their stability seems to be very hard in some cases. There
are no black-box tools applicable in general. 

The following elementary lemma is sometimes useful when considering
the asymptotic behaviour of a dynamical system. For completeness, a
proof is given in the Appendix.

\begin{lemma}
\label{odelemma}
Let $a$ and $b$ be Lipschitz continuous functions
$[0,\infty)\to(0,\infty)$. The unique solution $u$ of the differential
equation
$$
\dot{u}_t=b_t-a_tu_t,\quad t\ge0,
$$
with initial condition $u_0\ge0$ is positive for every $t>0$ and satisfies
$$
\frac{\liminf_{t\to\infty}b_t}{\limsup_{t\to\infty}a_t}\le
\liminf_{t\to\infty} u_t\le\limsup_{t\to\infty} u_t
\le\frac{\limsup_{t\to\infty}b_t}{\liminf_{t\to\infty}a_t}
$$
whenever the fractions are well-defined.
\end{lemma}

\section{Models and results}\label{modelsec}

\subsection{Plain Random Contact system}\label{plainrcsec}

Our first and simplest model is defined in Figure
\ref{plainrcpic}. The number of non-seed peers with chunk 0 (1) is
denoted by $X$ ($Y$). Peers arrive according to a Poisson process
with parameter $\lambda$, make a random contact, download whatever
chunk the contacted peer has (if the seed was contacted, the
downloaded chunk is chosen randomly), then make repeated random
contacts at Poisson rate 1 until the remaining chunk is found, and
leave the system. The system relies entirely on randomness, with fatal
consequences.

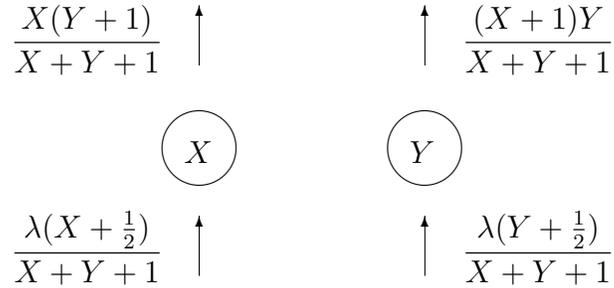
\begin{figure}[ht]
\begin{center}
\begin{minipage}{9.5cm}
\begin{picture}(90,40)
\put(30,20){\circle{10}}
\put(28,18){\makebox{$X$}}
\put(30,3){\vector(0,1){8}}
\put(5,5){\makebox{$\displaystyle{\frac{\lambda(X+\frac12)}{X+Y+1}}$}}
\put(30,31){\vector(0,1){8}}
\put(5,33){\makebox{$\displaystyle{\frac{X(Y+1)}{X+Y+1}}$}}
\put(60,20){\circle{10}}
\put(58,18){\makebox{$Y$}}
\put(60,3){\vector(0,1){8}}
\put(65,5){\makebox{$\displaystyle{\frac{\lambda(Y+\frac12)}{X+Y+1}}$}}
\put(60,31){\vector(0,1){8}}
\put(65,33){\makebox{$\displaystyle{\frac{(X+1)Y}{X+Y+1}}$}}
\end{picture}
\end{minipage}
\end{center}
\caption{Plain Random Contact system}
\label{plainrcpic}
\end{figure}

The large system limit of the Plain Random Contact system is the
dynamical system defined by the non-linear ordinary differential equations
\begin{equation}
\label{detplaindef}
\dot{x}=\frac{(\lambda-y)x}{x+y}\quad\dot{y}=\frac{(\lambda-x)y}{x+y}
\end{equation}
This system limit is easily seen to be unstable even close to its
equilibirium $x=y=\lambda$:

\begin{proposition}\label{plain-unstable-prop}
In $(0,\infty)^2$, the system (\ref{detplaindef}) has a single
equilibrium $(x^*,y^*)=(\lambda,\lambda)$. The system is not stable in
any environment of $(x^*,y^*)$. Starting with $x_0>\lambda>y_0$ we
have $x_t\to\infty$ and $y_t\to0$, and vice versa.
\end{proposition}
\begin{proof}
This is immediate from the equations (\ref{detplaindef}). 
\end{proof}

As long as both $x$ and $y$ are large and roughly of the same size,
the system empties rapidly. However, it is rather straightforward to
prove that the stochastic system is unstable when $\lambda$
is larger than one. 

{\em Remark:} In this paper, we are not interested in the possible
stability of the system when the input rate $\lambda$ is sufficiently
low. In all our models, $\lambda$ appears in the large scale limits as
a pure scaling parameter that can be as well chosen to be one. The
stability of the stochastic system may, however, depend on
$\lambda$. Susitaival and Aalto \cite{susitaivalaalto07} study by
simulations several two-chunk systems also from the point of view of
stability regions in terms of $\lambda$.

\begin{proposition}
With $\lambda>1$, the (stochastic) Plain Random Contact
system is unstable. More exactly, almost surely either $X$ or $Y$ escapes to
infinity whereas the other obtains ultimately only the values 0 and 1. 
\end{proposition}

\begin{proof}
We couple $(X,Y)$ with a process $(\tilde{X},\tilde{Y})$ such that
\begin{enumerate}
\item\label{couplineqs}
$\pr{\forall t\ge0\ X_t\ge\tilde{X}_t,\,Y_t\le\tilde{Y}_t}>0$,
\item\label{Xbig}
$\pr{\lim_{t\to\infty}\tilde{X}_t=\infty}=1$, and
\item\label{Ysmall}
$\pr{\tilde{Y}_t\in\set{0,1}\mbox{ ultimately}}=1$.
\end{enumerate}
Assume that $\lambda>1$ and choose a number $a\in(0,\lambda-1)$. Let
$\tilde{X}$ and $\tilde{Y}$ be mutually dependent, inhomogeneous birth
and death processes with up-jump and down-jump intensities defined as
follows, respectively:
$$
\alpha^{\tilde{Y}}_t=\frac{3\lambda}{2a(t\vee1)},
\quad\beta^{\tilde{Y}}_t=\tilde{Y}_t;
\quad\quad\alpha^{\tilde{X}}_t=a+1,
\quad\beta^{\tilde{X}}_t=\ \tilde{Y}_t+1.
$$
We have
$$
\pr{\tilde{Y}_t\ge1}\le\E{\tilde{Y}_t}
=\int_0^t\alpha^{\tilde{Y}}_se^{-(t-s)}\D{s}.
$$
Dividing the integration interval into sub-intervals $(0,1]$,
$(1,t-2\log t]$ and $(t-2\log t,t]$ we obtain the upper bound
\begin{equation}
\label{Ytildeub}
\pr{\tilde{Y}_t\ge1}
\le\frac{3\lambda}{2a}\left(e^{-t}+\frac{1}{t^2}+\frac{1}{t-2\log t}\right)
\sim\frac{3\lambda}{2at}.
\end{equation}
Let $N$ denote the counting process of the transitions of $\tilde{Y}$ from 1
to 2. Inequality (\ref{Ytildeub}) yields that there are a.s.\ only
finitely many such transitions. Indeed,
$$
\E{N_\infty}=\E{\int_0^\infty\alpha^{\tilde{Y}}_t\indic{\tilde{Y}_t=1}\D{t}}
\le\int_0^\infty\alpha^{\tilde{Y}}_t\pr{\tilde{Y}_t\ge1}\D{t}<\infty,
$$
since the rightmost integrand is $O(t^{-2})$. Thus, $\tilde{Y}$ has
the property \ref{Ysmall}. Since
$$
\lim_{t\to\infty}\frac{1}{t}\int_0^t\tilde{Y}_s\D{s}=0<a,
$$
we also obtain \ref{Xbig}. 

Now, choose $M\in\Nat$ so large that
$$
\frac{\lambda}{\displaystyle1+\frac{2}{M}}>a+1,
$$ 
set $X_0=\tilde{X}_0\ge M$, $Y_0=\tilde{Y}_0=0$ and define
$$
\tau=\inf\eset{t}{\tilde{X}_t<at\mbox{ or }\tilde{Y}_t>1}.
$$
Then obviously $\pr{\tau=\infty}>0$, and on $\set{\tau=\infty}$ we can
couple $(X,Y)$ with $(\tilde{X},\tilde{Y})$ in the desired way thanks
to domination relations between the respective intensities. Hence,
there exists a positive number $p$ such that
\begin{equation}
\label{posescape}
\cPr{\lim_{t\to\infty}X_t=\infty,\ Y_t\in\set{0,1}\mbox{
ultimately}}{X_0=x_0, Y_0=0}>p>0
\end{equation}
when $x_0\ge M$. By symmetry, the corresponding relation holds if $X$
and $Y$ are interchanged.

Next, let $R>(3/2)\lambda$ and note that
$$
\pr{(X_t,Y_t)\in[R\vee M,\infty)\times[R\vee M,\infty)\mbox{
ultimately}}=0.
$$
Indeed, when $X_t\wedge Y_t\ge R$, the total output rate of the system
is larger than the total input rate $\lambda$:
$$
\frac{2(X_t\wedge Y_t)}{\displaystyle1+\frac{X_t\wedge Y_t+1}{X_t\vee Y_t}}
>\frac{2}{3}R>\lambda.
$$
Thus, $\set{X_t\wedge Y_t\le R\vee M}$ is a recurrent event. By
inspecting the intensities depicted in Figure \ref{plainrcpic} it is
easy to see that the probability of moving from a state with $X\wedge
Y\le R\vee M$ to a state with $X\wedge Y=0$ before a change in $X\vee
Y$ is bounded from below by a positive constant not depending on the
value of $X\vee Y$. It follows that
$$
\set{X_t\vee Y_t\ge R\vee M,\ X_t\wedge Y_t=0}
$$
is a recurrent event as well. Now, (\ref{posescape}) yields the proposition.
\end{proof}

\subsection{Deterministic First Chunk system}\label{dfcsec}

Our first attempt to overcome the spontaneous imbalance tendency of
the Plain Random Contact system was the Deterministic Last Chunk
mechanism introduced in \cite{norprarei06}, where each peer decides in
advance (randomly) which chunk it will download as the last one. The
idea was to prevent peers from downloading systematically the rarest
chunk as the last one before leaving the system. In the two-chunk
case, defined by Figure \ref{dfcfig}, it might be more natural to
speak about the Deterministic First Chunk system. The number of empty
peers determined to download chunk 0 (1) first is denoted by $A$
($B$), while $X$ and $Y$ have their previous meaning.

\begin{figure}[ht]
\begin{center}
\begin{minipage}{9.5cm}
\begin{picture}(90,70)
\put(30,20){\circle{10}}
\put(28,18){\makebox{$A$}}
\put(30,3){\vector(0,1){8}}
\put(20,5){\makebox{$\displaystyle{\frac{\lambda}{2}}$}}
\put(30,31){\vector(0,1){8}}
\put(5,33){\makebox{$\displaystyle{\frac{A(X+1)}{X+Y+1}}$}}
\put(60,20){\circle{10}}
\put(58,18){\makebox{$B$}}
\put(60,3){\vector(0,1){8}}
\put(65,5){\makebox{$\displaystyle{\frac{\lambda}{2}}$}}
\put(60,31){\vector(0,1){8}}
\put(65,33){\makebox{$\displaystyle{\frac{B(Y+1)}{X+Y+1}}$}}
\put(30,50){\circle{10}}
\put(28,48){\makebox{$X$}}
\put(30,61){\vector(0,1){8}}
\put(5,63){\makebox{$\displaystyle{\frac{X(Y+1)}{X+Y+1}}$}}
\put(60,50){\circle{10}}
\put(58,48){\makebox{$Y$}}
\put(60,61){\vector(0,1){8}}
\put(65,63){\makebox{$\displaystyle{\frac{(X+1)Y}{X+Y+1}}$}}
\end{picture}
\end{minipage}
\end{center}
\caption{Deterministic First Chunk system}
\label{dfcfig}
\end{figure}
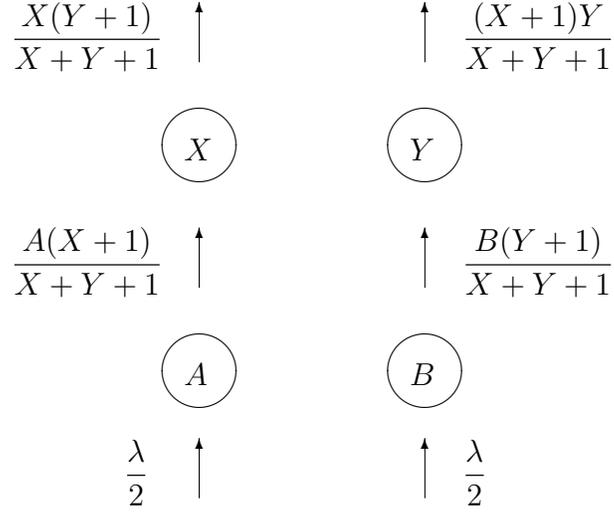

Although this balancing rule worked promisingly well in our
flash-crowd setup with many chunks, the two-chunk system is probably
unstable --- at least its large system limit 
\begin{eqnarray}
\label{detdfcdef}
\dot{a}=\frac{\lambda}{2}-\frac{ax}{x+y}
&\quad&\dot{b}=\frac{\lambda}{2}-\frac{by}{x+y}\\
\nonumber
\dot{x}=\frac{(a-y)x}{x+y}
&\quad&\dot{y}=\frac{(b-x)y}{x+y}
\end{eqnarray}
is unstable. 

\begin{proposition}
The system (\ref{detdfcdef}) has no unique equilibrium --- its
equilibria form the unbounded curve
$$
(a(\theta),b(\theta),x(\theta),y(\theta))
=(\theta,\frac{\lambda\theta}{2\theta-\lambda},
\frac{\lambda\theta}{2\theta-\lambda},\theta),
\quad\theta\in(\frac{\lambda}{2},\infty).
$$
Moreover, it has an open set of trajectories where two components
(either $x$ and $b$ or $y$ and $a$) grow to infinity while the other
two remain bounded.
\end{proposition}
\begin{proof}
Choose the initial values so that
$$
x(0)>b(0)>\lambda,\quad y(0)<\frac12\lambda,\quad
a(0)>\frac{(x(0)+y(0))\lambda}{2x(0)}.
$$
Then the same relations hold for the whole paths, as seen, with help
of Lemma \ref{odelemma}, by writing
\begin{eqnarray*}
\dot{a}&=&\frac{x}{x+y}\left(\frac{(x+y)\lambda}{2x}-a\right),\\
\dot{b}&=&\frac{y}{x+y}\left(\frac{\lambda}{2}
  +\frac{\lambda}{2y}\cdot x-b\right),\\
\dot{x}-\dot{b}&=&\frac{y}{x+y}\left(\frac{x}{y}\left(
a-\frac{(x+y)\lambda}{2x}\right)-(x-b)\right).
\end{eqnarray*}
Moreover, $x$ and $b$ grow toward infinity, $y$ decreases and $a$
approaches the value $\lambda/2$.
\end{proof}

We don't have a proof for the stochastic case, but on the basis of
similarity in behavior to the Plain Random Contact system (and
supported by some simulations), we conjecture that this system be
unstable as well:

\begin{conjecture}
With $\lambda$ large enough, the (stochastic) Deterministic First Chunk
system is unstable --- a.s., either $X$ or $Y$ escapes to infinity,
together with $B$ or $A$, respectively. 
\end{conjecture}

However, we don't fix a conjecture about a Deterministic Last Chunk
system with three or more chunks.

\subsection{Friedman system}\label{friedmansec}

Consider an urn containing balls with two colors. The simplest version
of Friedman's urn \cite{friedman49} works so that one repeatedly picks
a random ball from the urn and returns it together with a ball of the
opposite color. The proportions of the two colors approach
$\frac12-\frac12$ \cite{freedman65}.

We now modify the Plain Random Contact system by assuming (without
bothering how this might be realised) that arriving peers make a
random contact and then enter the system with a copy of the chunk that
the contacted peer did {\em not} have (in the case that the seed was
contacted, the downloaded chunk is chosen randomly); we call this
'complementary' random input 'Friedman input'. This system is
defined by Figure \ref{stochfriedmandef}.

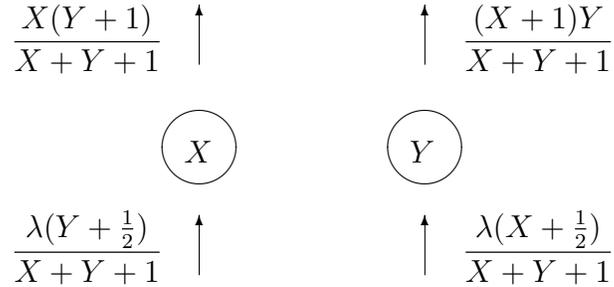
\begin{figure}[ht]
\begin{center}
\begin{minipage}{9.5cm}
\begin{picture}(90,40)
\put(30,20){\circle{10}}
\put(28,18){\makebox{$X$}}
\put(30,3){\vector(0,1){8}}
\put(5,5){\makebox{$\displaystyle{\frac{\lambda(Y+\frac12)}{X+Y+1}}$}}
\put(30,31){\vector(0,1){8}}
\put(5,33){\makebox{$\displaystyle{\frac{X(Y+1)}{X+Y+1}}$}}
\put(60,20){\circle{10}}
\put(58,18){\makebox{$Y$}}
\put(60,3){\vector(0,1){8}}
\put(65,5){\makebox{$\displaystyle{\frac{\lambda(X+\frac12)}{X+Y+1}}$}}
\put(60,31){\vector(0,1){8}}
\put(65,33){\makebox{$\displaystyle{\frac{(X+1)Y}{X+Y+1}}$}}
\end{picture}
\end{minipage}
\end{center}
\caption{Friedman system}
\label{stochfriedmandef}
\end{figure}

The corresponding large system limit is
\begin{equation}
\label{detfriedmandef}
\dot{x}=\frac{(\lambda-x)y}{x+y}\quad\dot{y}=\frac{(\lambda-y)x}{x+y}
\end{equation}

\begin{proposition}\label{detfriedmanprop}
The system (\ref{detfriedmandef}) has a single equilibrium
$(x^*,y^*)=(\lambda,\lambda)$ which is globally stable.  
\end{proposition}
\begin{proof}
The equations tell immediately that with any initial state in
$(0,\infty)^2$, both $x$ and $y$ converge monotonically to
$(\lambda,\lambda)$.
\end{proof}

Moreover, we note that $(x-\lambda)^2+(y-\lambda)^2$ is monotonically
decreasing:
$$
\frac{\D}{\D{t}}((x-\lambda)^2+(y-\lambda)^2)
=-\frac{2}{x+y}\left((x-\lambda)^2y+(y-\lambda)^2x\right)<0.
$$
This observation of a simple Lyapunov function can also be used for proving
the stability of the stochastic Friedman system:

\begin{proposition}\label{stochfriedmanprop}
The Friedman system defined in Figure \ref{stochfriedmandef} is stable
for any input rate $\lambda$.
\end{proposition}
\begin{proof}
Denote by $N^{x+}$ and $N^{x-}$ the counting processes of the up- and
down-jumps of the process $X$, respectively, so that
$X_t=X_0+\int_0^t(\D{N^{x+}_s}-\D{N^{x-}_s})$. We can then write
\begin{eqnarray*}
X^2_t-X^2_0&=&\int_0^t((X_{s-}+1)^2-X_{s-}^2)\D{N^{x+}_s}
  -\int_0^t((X_{s-}-1)^2-X_{s-}^2)\D{N^{x-}_s}\\
&=&\int_0^t(2X_{s-}+1)\D{N^{x+}_s}-\int_0^t(2X_{s-}-1)\D{N^{x-}_s}.
\end{eqnarray*}
(Note that $N^{x-}$ does not jump when $X_{s-}=0$.) The compensators
of $N^{x+}$ and $N^{x-}$ are, respectively,
$$
A^{x+}_t=\int_0^t\frac{\lambda(Y_s+\frac12)}{X_s+Y_s+1}\D{t},\quad
A^{x-}_t=\int_0^t\frac{X_s(Y_s+1)}{X_s+Y_s+1}\D{t}.
$$
Using similar notation for the process $Y$, we see that the process
$X^2_t+Y^2_t$ is compensated to a martingale by subtracting from it
the process $A_t=\int_0^ta_s\D{t}$, where
$$
a=\left(4\lambda(X+\frac12)(Y+\frac12)-2X(X-\frac12)(Y+1)-2(X+1)(Y-\frac12)Y\right)/(X+Y+1).
$$
Denoting $m=X\wedge Y$, $M=X\vee Y$, we have the estimate
\begin{eqnarray*}
a&=&\frac{1}{M+m+1}\left(4\lambda(M+\frac12)(m+\frac12)
  -2M(M-\frac12)(m+1)-2(M+1)m(m-\frac12)\right)\\
&\le&\frac{2M(m+1)}{M+m+1}\left(2\lambda(1+\frac{1}{2M})-M+\frac12\right)\\
&\le&\frac12(3\lambda+\frac12-M)\\
&\le&-\frac14
\end{eqnarray*}
on the set $\set{M>3\lambda+1}$. 
Denote $\tau=\inf\eset{t}{M_t\le3\lambda+1}$. Assume now that
$M_0=X_0\vee Y_0>3\lambda+1$. Then, the process
$B_t=X^2_{t\wedge\tau}+Y^2_{t\wedge\tau}$ is a
non-negative supermartingale and thus converges to an integrable limit
satisfying $\E{B_\infty}\le\E{B_0}$, and $B_t-A_{t\wedge\tau}$ is a
martingale (for integrability, note that $B_t$ is dominated by
$(X_0+Y_0+N^{x+}_t+N^{y+}_t)^2$, where $N^{x+}_t+N^{y+}$ is a Poisson
process). Since $A_{t\wedge\tau}$ is non-increasing, we have
$$
B_0\ge
B_0-\lim_{t\to\infty}\E{B_t}=-\lim_{t\to\infty}\E{A_{t\wedge\tau}}
\ge\E{\int_0^\tau\frac14\D{t}}=\frac14\E{\tau}.
$$
Thus, the finite set $\eset{(x,y)\in\Nat^2}{x\vee y\le3\lambda+1}$ is
reached from any fixed initial state outside of it in a time having
finite expectation. 
\end{proof}

\subsection{Delayed Friedman system}\label{delfriedsec}

Our first distributed implementation of the idea of the Friedman
system (see \cite{reittunorrosphyscomnet}) was the system defined by Figure
\ref{delfriedmanpic}. An arriving peer first makes a random contact
and then {\em decides} to download first the chunk that the contacted
peer does {\em not} have (in the case that the seed was contacted, the
peer decides randomly). As with the Deterministic First Chunk system,
the number of empty peers determined to download chunk 0 (1) first is
denoted by $A$ ($B$), while $X$ and $Y$ have the same meaning as
before. We call this the Delayed Friedman system, because the
subsystem $(X,Y)$ obtains Friedman input with stochastic delay.

\begin{figure}[ht]
\begin{center}
\begin{minipage}{9.5cm}
\begin{picture}(90,70)
\put(30,20){\circle{10}}
\put(28,18){\makebox{$A$}}
\put(30,3){\vector(0,1){8}}
\put(5,5){\makebox{$\displaystyle{\frac{\lambda(Y+\frac12)}{X+Y+1}}$}}
\put(30,31){\vector(0,1){8}}
\put(5,33){\makebox{$\displaystyle{\frac{A(X+1)}{X+Y+1}}$}}
\put(60,20){\circle{10}}
\put(58,18){\makebox{$B$}}
\put(60,3){\vector(0,1){8}}
\put(65,5){\makebox{$\displaystyle{\frac{\lambda(X+\frac12)}{X+Y+1}}$}}
\put(60,31){\vector(0,1){8}}
\put(65,33){\makebox{$\displaystyle{\frac{B(Y+1)}{X+Y+1}}$}}
\put(30,50){\circle{10}}
\put(28,48){\makebox{$X$}}
\put(30,61){\vector(0,1){8}}
\put(5,63){\makebox{$\displaystyle{\frac{X(Y+1)}{X+Y+1}}$}}
\put(60,50){\circle{10}}
\put(58,48){\makebox{$Y$}}
\put(60,61){\vector(0,1){8}}
\put(65,63){\makebox{$\displaystyle{\frac{(X+1)Y}{X+Y+1}}$}}
\end{picture}
\end{minipage}
\end{center}
\caption{Delayed Friedman system.}
\label{delfriedmanpic}
\end{figure}
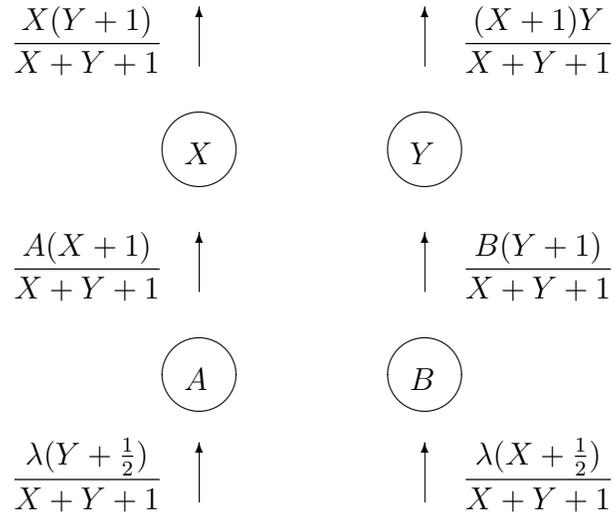

The corresponding large system limit is the dynamical system
\begin{eqnarray}
\label{detdelfried}
\dot{a}=\frac{\lambda y-ax}{x+y}&&\dot{b}=\frac{\lambda x-by}{x+y}\\
\nonumber
\dot{x}=\frac{(a-y)x}{x+y}&&\dot{y}=\frac{(b-x)y}{x+y}.
\end{eqnarray}
This system is difficult to analyse, because it oscillates heavily
around its unique equilibrium
$(a^*,b^*,x^*,y^*)=(\lambda,\lambda,\lambda,\lambda)$.  The 'logic' of
the oscillating system evolution from an imbalanced state, depicted in Figure
\ref{estrffig}, is the following:
\begin{itemize}
\item there is hardly any input to nor output from $y$ 
\item the 'Friedman rule' directs input to $a$, which accumulates
almost all of it, since $x$ is negligible
\item when enough mass has accumulated to $a$, the balance starts to improve 
\item once $x$ has become macroscopic, $a$ and $y$ empty rapidly
\item $b$ has not had time to grow, so we get a situation close to the mirror
image of the original.
\end{itemize}

\begin{figure}[ht]
\begin{center}
\includegraphics[width=10cm]{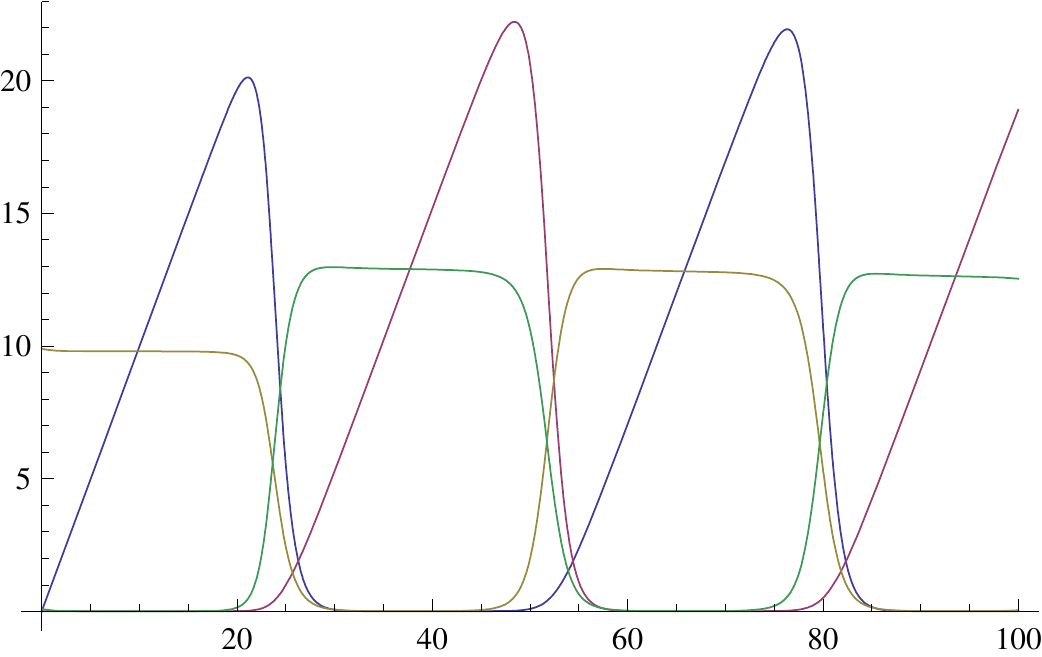}\\
\end{center}
\caption{A trajectory of the large system limit
(\protect\ref{detdelfried}) of the Delayed Friedman system with
initial values $a=b=0$, $x=0.1$, $y=9.9$. Colors: $a$: blue, $b$: red,
$x$: green, $y$: brown.}
\label{estrffig}
\end{figure}

Numerical experiments like that shown in Figure \ref{rhodown12} (left)
hint to global stability, but we have not found an explicit Lyapunov
function or other means to prove this. As regards local behavior near
equilibrium, the linearised system at equilibrium is essentially a
two-dimensional harmonic oscillator --- the two remaining dimensions
correspond to negative eigenvalues. If we continue the numerical
computation of a trajectory of (\ref{detdelfried}), we see slow
convergence toward equilibrium (Figure \ref{rhodown12}, right; by the
way, it is interesting to observe that $x+y$ seems to never descend
below 2). Indeed, the system (\ref{detdelfried}) turns out to be
locally asymptotically stable, and this can be shown in a
non-elementary but basically straightforward way through a center
manifold analysis (see \cite{chha:metbif,Kuznetsov:applbifurc}).

\begin{figure}[ht]
\begin{center}
\hspace*{\fill}\includegraphics[width=6cm]{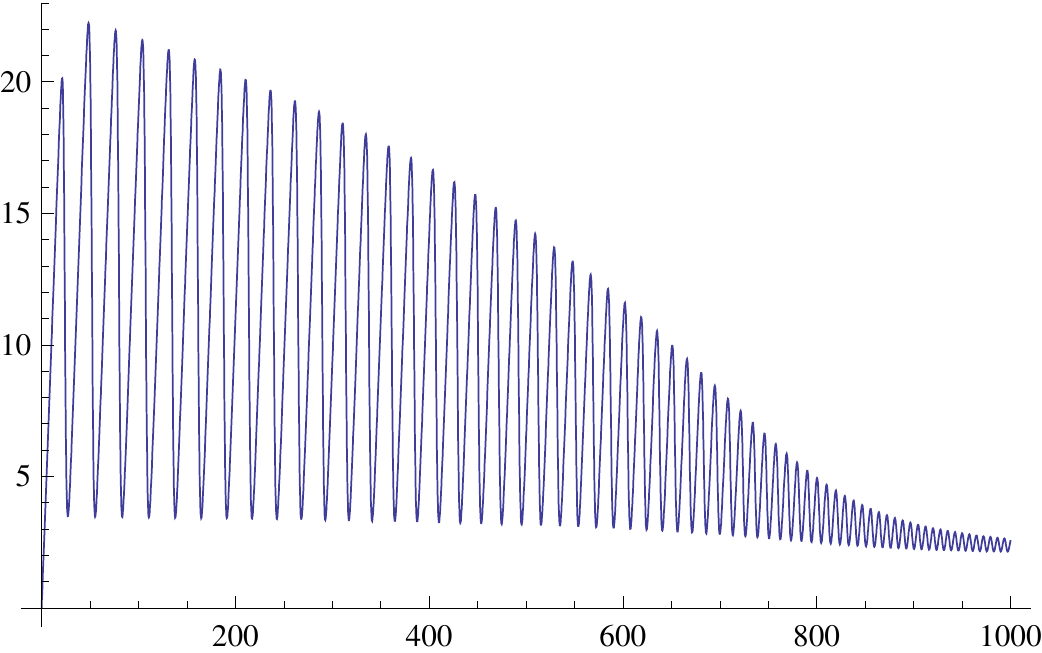}
\hspace{\fill}
\includegraphics[width=6cm]{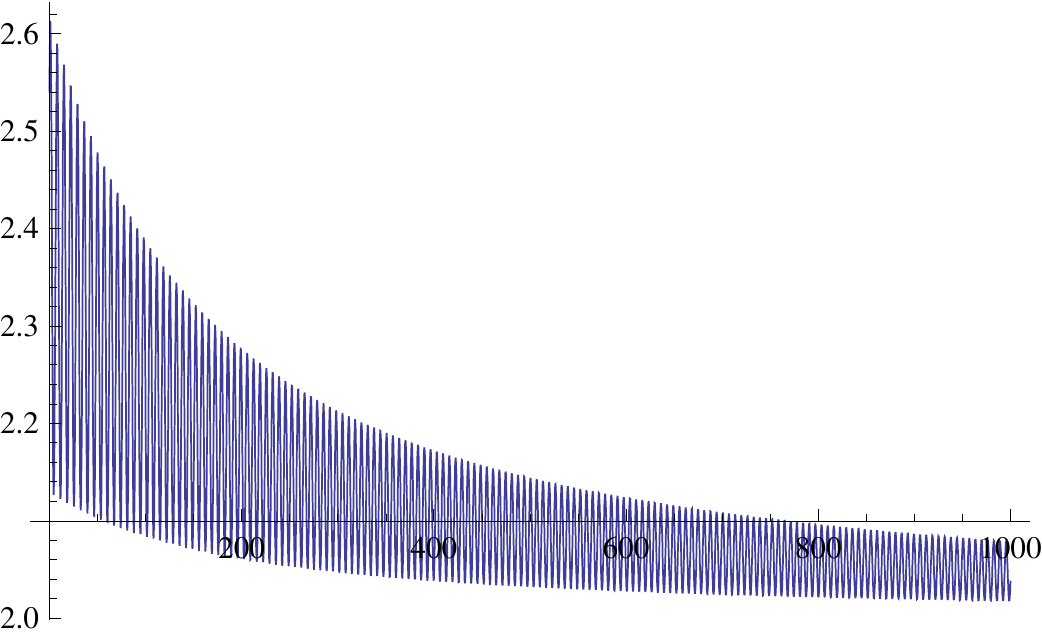}
\hspace*{\fill}\\
\end{center}
\caption{The behaviour of $x+y$, initial state as in Figure
\protect{\ref{estrffig}}. Left: $t\in[0,1000]$. Right:
$t\in[1000,2000]$.}
\label{rhodown12}
\end{figure}

\begin{proposition}
The system defined by (\ref{detdelfried}) is locally asymptotically
stable. If the starting point $(a_0,b_0,x_0,y_0)$ is close enough to
the equilibrium $(\lambda,\lambda,\lambda,\lambda)$, the distance to
it decreases proportionally to $1/\sqrt{t}$. 
\end{proposition}
\begin{proof}
Since the common denominator of the right hand sides of
(\ref{detdelfried}) does not affect the trajectories, and since
$\lambda$ is a pure scaling factor, it is sufficient to consider the
system
\[\x'=\bmat{x_4-x_1\,x_3\\x_3-x_2\,x_4\\(x_1-x_4)\,x_3\\(x_2-x_3)\,x_4}
\itext{with equilibrium}\p=\smat{1\\1\\1\\1}\ .\] 
At $\p$, the
linearized system is given by the matrix
\[\A=\bmat{-1&0&-1&1\\0&-1&1&-1\\1&0&0&-1\\0&1&-1&0}\ .\]
This has eigenvalues $\,\lambda_{1,2}=\pm i\,$ and $\,\lambda_{3,4}=-1\,$.
The latter has just one linearly independent eigenvector. 
Let $\,\v_1\,$ and $\,\v_2\,$ be the
real and imaginary parts of an eigenvector corresponding to $\,\lambda_1=i\,$ 
and let $\,\v_3\,$ be an eigenvector for $\,\lambda_3=-1\,$ and $\,\v_4\,$
such that $\,\A\,\v_4+\v_4=\v_3\,$. Putting these into matrix $\,\V\,$ 
we obtain the similarity transformation
\[\V^{-1}\,\A\,\V=\bmat{0&1&0&0\\-1&0&0&0\\0&0&-1&1\\0&0&0&-1}\itext{with}
\V=\bmat{1&1&0&1\\-1&-1&0&1\\-1&0&1&0\\1&0&1&0}\ .\] Thus, the system
has a two-dimensional stable manifold $\,W^s\,$ and a two-dimensional
center manifold $\,W^c\,$.  The corresponding columns of $\,\V\,$ are
tangents to these at $\p$ (see e.g. \cite{chha:metbif}), and the
manifolds are invariant under the flow of the system. Trajectories
close to $\,\p\,$ approach the center manifold like
$\,t\,e^{-t}\,$. By the reduction principle (see
\cite{Kuznetsov:applbifurc}), in order to study the stability of the
system, it suffices to know how it behaves on the local center
manifold. For this we will reduce the system to a normal form (see
\cite{Kuznetsov:applbifurc}).

Take first a linear change of coordinates:
\[\x=\p+\V\,\smat{\u\\\v}\,,\itext{where}\u,\v\in\Re^2\ .\]
Then the local center manifold can be expressed as
\[W^c=\eset{\p+\V\,\smat{\u\\\h(\u)}}{\norm{\u}<d}\ ,\]
where $\,\h(\u)\,$ is a two-vector with $\,\h(0,0)=0\,$ and $\,D\h(0,0)=0\,$
(the latter because of tangency).

Write $\,\smat{\f^c\\\f^s}=\V^{-1}\,\f\,$. Then the requirement of 
invariance of $\,W^c\,$ amounts to
\[\f^s(\p+\V\,\smat{\u\\\h(\u)})=
D\h(\u)\,\f^c(\p+\V\,\smat{\u\\\h(\u)})\ .\]
This is the equation for $\,\h\,$. We can recursively solve the coefficients 
of the Taylor expansion of $\,\h\,$. Note that the constant and linear
terms are zero. This way we get:
\[\h(\u)=\tfrac19\,\bmat{5\,u_1^2+u_1u_2-u_2^2\\
  -5\,u_1^2-14\,u_1u_2-2\,u_2^2}+O(\norm\u^3)\]
and 
\[\g(\u):=\f^c(\p+\V\,\smat{\u\\\h(\u)})=
  \bmat{-u_2+(-10\,u_1^3+8\,u_1^2u_2-2\,u_1u_2^2+u_2^3\,)/9\,\\u_1}
  +O(\norm\u^4)\ .\]
The behaviour of the system on the center manifold is given by equation
$\,\u'=\g(\u)\,.$ 
Following \cite{Kuznetsov:applbifurc} we change to complex variables. 
An eigenvector of the linear part of $\,\g\,$ corresponding to the eigenvalue
$\,i\,$ is $\,\q=\smat{i\\1}\,$. Setting $\,\u=z\,\q+\bar z\,\bar\q\,$  
we get $\,z=(u_2-i\,u_1)/2\,$ and
\[z'=i\,z+\tfrac{8+7i}{18}\,z^3+\tfrac{32+11i}{18}\,z^2\,\bar z
+\tfrac{32-11i}{18}\,z\,\bar z^2+\tfrac{-8+7i}{18}\,\bar z^3
+O(\abs z^4)\ .\]
Taking substitution $\,z=w+\beta_3\,w^3+\beta_2\,w^2\,\bar w+\beta_1\,w\,\bar w^2
+\beta_0\,\bar w^3 \,$ we obtain
\[w'=i\,w+(\tfrac{8+7i}{18}-2\,\beta_3)\,w^3-\tfrac{32+11i}{18}\,w^2\,\bar w
+(\tfrac{32-11i}{18}+2i\,\beta_1)\,w\,\bar w^2+(\tfrac{-8+7i}{18}+4i\,\beta_0)\,
\bar w^3+O(\abs w^4)\ .\]
We see that we can kill other third order terms except the $\,w^2\,\bar w\,$-term.
Hence choosing
\[\beta_0=\tfrac{7-8i}{72}\,,\ \beta_1=\tfrac{11+32i}{36}\,,\ 
\beta_2=0\,,\ \beta_3=\tfrac{7-8i}{36}\,,\] 
we obtain the equation
\[w'=i\,w-\tfrac{32+11i}{18}\,\abs w^2\,w+O(\abs w^4)\ .\]

Now, since
\[\tfrac{d\ }{dt}\abs w^2=     
(i\,w-\tfrac{32+11i}{18}\,\abs w^2\,w)\,\bar w
+w\,(-i\,\bar w-\tfrac{32-11i}{18}\,\abs w^2\,\bar w)+O(\abs w^5)
=-\tfrac{32}{9}\abs w^4+O(\abs w^5)\ ,\]
we see that $\,w\to 0\,$ like $\,1/\sqrt t\,$. Hence the system is
locally asymptotically stable. 
\end{proof}

If the third order terms of the equation of $\,w\,$ would not have
determined the stability, we should have gone to higher order expansions.

Our simulations showed that also the original stochastic system
oscillates strongly, with a roughly constant amplitude depending on
$\lambda$, even when it is started from a balanced state. The
simulations suggested that the stochastic system might be stable, but
we have not found a way to prove or disprove this. From a practical
viewpoint, the oscillations indicate that the system is not in good
balance, which could be fatal in an application with rapidly changing
arrival rates and flash crowd scenarios.

\subsection{Enforced Friedman system}\label{forcedfriedmansec}

The reason of the oscillatory behaviour of the Delayed Friedman system
was that when a peer finally succeeds in downloading its first chunk
according to the 'Friedman rule', applied possibly long ago, the
choice may already be outdated and thus counterproductive. Our last
model avoids this problem by realising choices only immediately or not
at all. An empty peer makes three contacts simultaneously (sampling
with replacement; the seed shows a random chunk). If two of the chunks
of the contacted peers differ from the third one, the latter is
downloaded. If all three chunks are similar, nothing is done, but the
empty peer stays in a waiting room and repeats the triple contact
operation after Exp(1)-distributed waiting time. The number of peers
in the waiting room is denoted by $Z$. Note that if the experiment was
successful (that is, a 2-1 situation was obtained), the probability of
downloading chunk 0 is $(Y+\frac12)/(X+Y+1)$ --- exactly the same as
in the Friedman system! Therefore we call this system, defined by
Figure \ref{forcefriedpic}, the Enforced Friedman system.

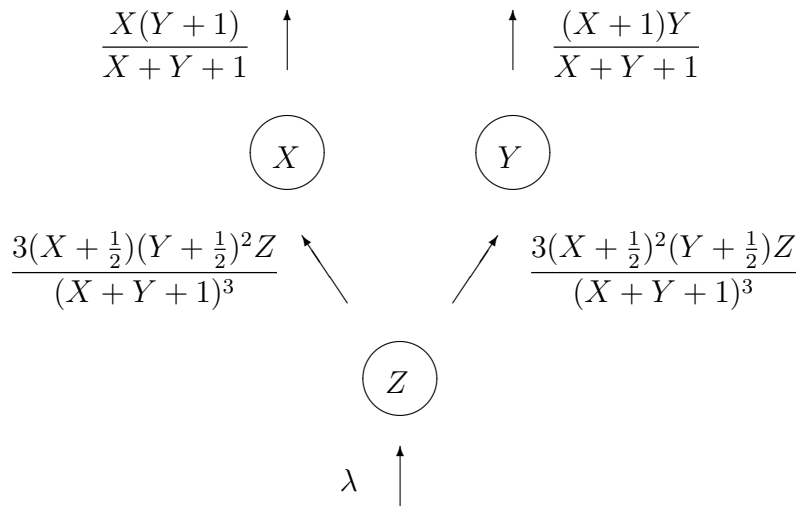
\begin{figure}[ht]
\begin{center}
\begin{minipage}{10.5cm}
\begin{picture}(105,70)(-7,0)
\put(45,20){\circle{10}}
\put(43,18){\makebox{$Z$}}
\put(45,3){\vector(0,1){8}}
\put(37,5){\makebox{$\lambda$}}
\put(38,30){\vector(-2,3){6}}
\put(-7,33){\makebox{$\displaystyle{\frac{3(X+\frac12)(Y+\frac12)^2Z}{(X+Y+1)^3}}$}}
\put(52,30){\vector(2,3){6}}
\put(62,33){\makebox{$\displaystyle{\frac{3(X+\frac12)^2(Y+\frac12)Z}{(X+Y+1)^3}}$}}
\put(30,50){\circle{10}}
\put(28,48){\makebox{$X$}}
\put(30,61){\vector(0,1){8}}
\put(5,63){\makebox{$\displaystyle{\frac{X(Y+1)}{X+Y+1}}$}}
\put(60,50){\circle{10}}
\put(58,48){\makebox{$Y$}}
\put(60,61){\vector(0,1){8}}
\put(65,63){\makebox{$\displaystyle{\frac{(X+1)Y}{X+Y+1}}$}}
\end{picture}
\end{minipage}
\end{center}
\caption{Enforced Friedman system.}
\label{forcefriedpic}
\end{figure}

The large system limit of the Enforced Friedman system is
\begin{eqnarray}
\label{detforcedfriedmandef}
\dot{z}&=&\lambda-\frac{3xyz}{(x+y)^2}\\
\label{xyz-eqs}
\dot{x}&=&\frac{3xy^2z}{(x+y)^3}-\frac{xy}{x+y}\\
\nonumber
\dot{y}&=&\frac{3x^2yz}{(x+y)^3}-\frac{xy}{x+y}.
\end{eqnarray}
This is trickier than the plain Friedman system, but nevertheless tractable:

\begin{proposition}\label{detforcedfriedmanprop}
The system (\ref{detforcedfriedmandef}) has in
$(0,\infty)\times(0,\infty)\times(0,\infty)$ a single equilibrium
point $(z^*,x^*,y^*)=(\frac43\lambda,\lambda,\lambda)$. Locally, the
equilibrium is a sink, i.e.\ all eigenvalues of the limiting linear
system are real and negative. This equilibrium is globally
asymptotically stable.
\end{proposition}

\begin{proof}
The uniqueness and local character of the equilibrium are found by
easy computations. It remains to prove the global stability. Since
$\lambda$ is a pure scaling parameter, we can choose $\lambda=1$. Let
us change to the variables
$$
z,\quad\rho=x+y,\quad\beta=\left(\frac{x-y}{x+y}\right)^2,
$$
which lose the differentiation between $x$ and $y$ in favour of a
single imbalance characteristic $\beta$. The new system satisfies the
equations
\begin{eqnarray}
\label{z-eq}
\dot{z}&=&1-\frac34(1-\beta)z\\
\label{rho-eq}
\dot\rho&=&\frac14(1-\beta)(3z-2\rho)\\
\label{beta-eq}
\dot\beta&=&-\beta(1-\beta)\left(\frac{3z}{\rho}-1\right).
\end{eqnarray}
Fix arbitrary initial values $z_0>0$, $\rho_0>0$, $\beta_0\in(0,1)$
($\beta_0=0$ gives a simple linear system).

$1^{\mathrm{o}}$. Note first that the system $(z_t,\rho_t,\beta_t)$
cannot escape from the set $(0,\infty)\times(0,\infty)\times[0,1)$ in
finite time. Denote $\tau=\inf\eset{t\ge0}{\beta_t=1}$. For
$t\in[0,\tau)$, $\dot{z}_t>0$ whenever $z_t<4/3$, and $\dot{\rho}_t>0$
whenever $\rho_t<(3/2)z_t$. Since the system freezes at $\tau$,
it follows that $\inf_{t\ge0}z_t\ge\min(z_0,4/3)$ and
$\inf_{t\ge0}\rho_t>0$. Equation (\ref{beta-eq}) then yields that
$\beta$ cannot reach 1 at any finite time, so $\tau=\infty$. 

$2^{\mathrm{o}}$. A basic observation from the equations
(\ref{xyz-eqs}) is that the absolute value of $x-y$ is
non-increasing. In terms of the variables $z,\rho,\beta$ this means
that $\beta\rho^2$ is non-increasing, and the corresponding equation
from which this is seen reads
\begin{equation}
\label{betarho2down}
\frac{\mathrm{d}}{\mathrm{d}t}(\beta_t\rho^2_t)=-\frac{3}{2}\beta_t(1-\beta_t)z_t.
\end{equation}
By point $1^{\mathrm{o}}$, we get 
\begin{equation}
\label{inttozero}
\lim_{t\to\infty}\int_t^\infty\beta_s(1-\beta_s)z_s\D{s}=0.
\end{equation}
Assume that $\beta_t\to1$. Since $\beta_t\rho_t^2$ is decreasing by
(\ref{betarho2down}), $\rho_t$ cannot be ultimately increasing, and it
approaches some finite limit value. On the other hand, Lemma
\ref{odelemma}, applied to (\ref{z-eq}), yields with $\beta_t\to1$
that $z_t\to\infty$, and (\ref{rho-eq}) makes $\rho_t$ ultimately
increasing. This contradiction shows that
$w:=1-\lim\inf_{t\to\infty}\beta_t>0$.

Since we already saw that ultimately $z_t\ge1$, we obtain by
(\ref{inttozero}) and (\ref{beta-eq}) (neglecting the term $3z/\rho$
that for any $\epsilon>0$ there exists a number $t_0$ such that
\begin{equation}
\label{epsineq-ff}
\epsilon>\int_t^\infty\beta_s(1-\beta_s)z_s\D{s}
\ge\int_t^T\dot\beta_s\D{s}=\beta_T-\beta_t
\end{equation}
for any $t$ and $T$ such that $T>t>t_0$. Choosing $\epsilon=w/2$ we
deduce that $\lim\sup_{t\to\infty}\beta_t<1$. But then,
(\ref{inttozero}) and (\ref{epsineq-ff}) imply
$\lim_{t\to\infty}\beta_t=0$. 

$3^{\mathrm{o}}$. Now, Lemma \ref{odelemma} yields that $z_t$ and $\rho_t$
converge to their stable values.
\end{proof}

In simulations, this scheme seems to work in a perfect way, and it is
also capable to deal with flash crowds and rapid input rate variations
as well (say, when $\lambda$ is large for a while and then suddenly
drops). We also have good intuitive grounds to believe that the
stochastic Enforced Friedman system be stable: as noted earlier, the
output from the waiting room is pure Friedman input to the subsystem
$(X,Y)$, just with a randomly varying rate. Since the Friedman system
is provably stable with any input rate $\lambda$, the same can be
expected from the Enforced Friedman system:

\begin{conjecture}
The (stochastic) Enforced Friedman system is stable with any input
rate $\lambda$.  
\end{conjecture}

\section{Concluding remarks}\label{conclsec}

Although we could so far prove only some of the
stability/non-stability properties that seem to hold for our models,
some preliminary conclusions can already be drawn. First, we showed
that it is possible to create effectively stabilizing 'Friedman input'
using a fully distributed algorithm based on random contacts. Second,
one should use a 'waiting room' concept with memoryless re-trials
rather than launch an active hunt for a rare chunk --- the rarity may
dissolve during the hunt, leading to oscillations. Third, the extra
delay caused by the waiting room stage is relatively small in the
mean, certainly worth its price. We also believe that the idea of
Enforced Friedman input has usable counterparts in many-chunk systems,
and we work currently on one promising candidate.

On the other hand, note that our models lack realism as regards the
modelling of two aspects of the system: search and
downloading. Following \cite{massoulievojnovic05}, the slowness of
downloading is modelled by the Exp(1)-distributed inter-contact times,
whereas the contacts themselves are assumed to be instant. However,
unsuccessful contacts could be repeated much faster than those leading
to a chunk download. Assessing the significance of this simplification
remains for further work.

\bigskip

{\bf Acknowledgement.} We thank Balakrishna Prabhu, Rudesindo Nunez
Queija, Philippe Robert, Florian Simatos (members of the Euro-FGI
SCALP project) and Lasse Leskel\"a for fruitful discussions and
insights.


\appendix

\section{Appendix}\label{auxsec}
\begin{proof} of Lemma \ref{odelemma}.

Denote $\tau=\inf\eset{t>0}{u_t=0}$. Since $\dot{u}_0=b_0>0$
in the case that $u_0=0$, we always have $\tau>0$. For $t\in[0,\tau)$ we have
$$
\dot{u}_t\ge\inf_{s\in[0,\tau)}b_s-(\sup_{s\in[0,\tau)}a_s)u_t. 
$$
It follows that
$$
u_t\ge\frac{\inf_{s\in[0,\tau)}b_s}{\sup_{s\in[0,\tau)}a_s}
+(u_0-\frac{\inf_{s\in[0,\tau)}b_s}{\sup_{s\in[0,\tau)}a_s})
\exp(-(\sup_{s\in[0,\tau)}a_s)t),\quad t\in[0,\tau).
$$
Since the right hand side is positive also for every $t=\tau$, we deduce that
$\tau=\infty$. 

Assume first that the functions $a$ and $b$ are bounded away from zero
and infinity. Let $a_*,a^*,b_*,b^*$ be any
positive numbers such that
$$
a_*<\liminf_{t\to\infty}a_t\le\limsup_{t\to\infty}a_t<a^*,\quad
b_*<\liminf_{t\to\infty}b_t\le\limsup_{t\to\infty}b_t<b^*.
$$
Denote
$$
\sigma=\inf\eset{t\ge0}{\inf_{s\in[t,\infty)}a_s\ge a_*,\
\sup_{s\in[t,\infty)}a_s\le a^*,\
\inf_{s\in[t,\infty)}b_s\ge b_*,\
\sup_{s\in[t,\infty)}b_s\le b^*}.
$$ 
Since $\sigma$ is finite and
$$
\inf_{s\in[\sigma,\infty)}b_s-(\sup_{s\in[\sigma,\infty)}a_s)u_{\sigma+t}
\le\dot{u}_{\sigma+t}
\le\sup_{s\in[\sigma,\infty)}b_s-(\inf_{s\in[\sigma,\infty)}a_s)u_{\sigma+t},
\quad t\ge0,
$$
we obtain that
$$
\frac{b_*}{a^*}+(u_\sigma-\frac{b_*}{a^*})
e^{-a^*t}
\le u_{\sigma+t}
\le\frac{b^*}{a_*}+(u_\sigma-\frac{b^*}{a_*})
e^{-a_*t},\quad t\ge0,
$$
which implies
$$
\frac{b_*}{a^*}
\le\liminf_{t\to\infty} u_t\le\limsup_{t\to\infty} u_t
\le\frac{b^*}{a_*}.
$$
Since $a_*,a^*,b_*,b^*$ were
arbitrary, the assertion follows.

Finally, when $a$ and $b$ are not both bounded away from zero and
infinity and the fractions in the assertion are well-defined, the
non-trivial cases are obtained as above, simply relaxing some
conditions in the definition of $\sigma$.
\end{proof}

\end{document}